\documentclass[envcountsame,envcountsect, oribibl]{llncs}
\usepackage{amsmath,amssymb}
\usepackage{bbold, pgf, tikz, verbatim,stmaryrd,eurosym}
\usepackage{times}

\usetikzlibrary{arrows,automata}

\DeclareMathOperator{\dom}{dom}
\DeclareMathOperator{\val}{val}
\DeclareMathOperator{\wt}{wt}
\DeclareMathOperator{\uWAL}{\bf WAL}
\DeclareMathOperator{\eWAL}{\bf eWAL}
\DeclareMathOperator{\MSO}{\bf MSO}
\DeclareMathOperator{\Const}{\text{\sc Const}}

\DeclareMathOperator{\Free}{\text{\sc Free}}

\DeclareMathOperator{\false}{\bf false}
\DeclareMathOperator{\Run}{Run}
\DeclareMathOperator{\code}{code}
\DeclareMathOperator{\Label}{label}

\newcommand{\A}{\mathcal A}
\newcommand{\N}{\mathbb N}
\newcommand{\zero}{\mathbb 0}
\newcommand{\one}{\mathbb 1}
\newcommand{\Ratio}{\text{\sc Ratio}}
\newcommand{\Disc}{\text{\sc Disc}}
\newcommand{\Energy}{\text{\sc Energy}}
\newcommand{\RR}{\overline{\mathbb R}}
\newcommand{\llangle}{\langle \! \! \, \langle}
\newcommand{\rrangle}{\rangle \! \! \, \rangle}

\begin{document}
\title{Weight Assignment Logic\thanks{This is the full version of the paper published at DLT 2015 \cite{Pe15}}}
\author{Vitaly Perevoshchikov\thanks{Supported by DFG Research Training Group 1763 (QuantLA)}}
\institute{Universit\"at Leipzig,  Institut f\"ur Informatik, \\ 
04109 Leipzig, Germany\\
\email{perev@informatik.uni-leipzig.de}
}

\maketitle

\begin{abstract}
We introduce a weight assignment logic for reasoning about quantitative languages of infinite words. 
This logic is an extension of the classical MSO logic and permits to describe quantitative properties of systems with multiple weight parameters, e.g., the ratio
between rewards and costs. We show that this logic is expressively equivalent to unambiguous weighted B\"uchi automata. We also consider an extension of weight assignment logic which is expressively equivalent to nondeterministic weighted B\"uchi automata. %The translation process from logics into automata and vice versa is effective.

\end{abstract}

\section{Introduction}

Since the seminal B\"uchi theorem \cite{Buc60} about the expressive equivalence of finite automata and monadic second-order logic, a significant field of research investigates logical characterizations of language classes appearing from practically relevant automata models. In this paper we introduce a new approach to the logical characterization of quantitative languages of infinite words where every infinite word carries a value, e.g., a real number. 

Quantitative languages of infinite words and various weighted automata for them were investigated by Chatterjee, Doyen and Henzinger in \cite{CDH08} as models for verification of quantitative properties of systems. Their weighted automata are automata with a single weight parameter where a computation is evaluated using measures like the limit average or discounted sum. Recently, the problem of analysis and verification of systems with multiple weight parameters, e.g. time, costs and energy consumption, has received much attention in the literature \cite{BJLLS12, BGHJ09, BBL08, FJLLS11, FGR12, LR05}. For instance, the setting where a computation is evaluated as the ratio between accumulated rewards and costs was considered in \cite{BGHJ09, BBL08, FGR12}. Another example is a model of energy automata with several energy storages \cite{FJLLS11}.

\medskip

{\bf Related work.} Droste and Gastin \cite{DG07} introduced weighted MSO logic on finite words with constants from a semiring. In the semantics of their logic (which is a quantitative language of finite words) disjunction is extended by the sum operation of the semiring and conjunction is extended by the product. They show that weighted MSO logic is more expressive than weighted automata \cite{DKV09} (the unrestricted use of weighted conjunction and weighted universal quantifiers leads to unrecognizability) and provide a syntactically restricted fragment which is expressively equivalent to weighted automata. This result was extended in \cite{DR06} to the setting of infinite words. A logical characterization of the quantitative languages of Chatterjee, Doyen and Henzinger was given in \cite{DM12} (again by a restricted fragment of weighted MSO logic). In \cite{DP14}, a multi-weighted extension of weighted MSO logic of \cite{DM12} with the multiset-based semantics was considered. 

\medskip

{\bf Our contributions.} In this paper, we introduce a new approach to logic for quantitative languages, different from \cite{DG07,  DM12, DP14, DR06}. We develop a so-called {\em weight assignment logic} (WAL) on infinite words, an extension of the classical MSO logic to the quantitative setting. This logic allows us to assign weights (or multi-weights) to positions of an $\omega$-word. Using WAL, we can, for instance, express that whenever a position of an input word is labelled by letter $a$, then the weight of this position is $2$. As a weighted extension of the logical conjunction, we use the {\em merging} of partially defined $\omega$-words. In order to evaluate a partially defined $\omega$-word, we introduce a {\em default weight}, assign it to all positions with undefined weight, and evaluate the obtained totally defined $\omega$-word, e.g., as the reward-cost ratio or discounted sum. 

As opposed to the weighted MSO logic of \cite{DG07}, the weighted conjunction-like operators of WAL capture recognizability by weighted B\"uchi automata. We show that WAL is expressively equivalent to {\em unambiguous} weighted B\"uchi automata where, for every input $\omega$-word, there exists at most one accepting computation. Unambiguous automata are of considerable interest for automata theory as they can have better decidability properties. For instance, in the setting of finite words, the equivalence problem for unambiguous max-plus automata is decidable \cite{HIJ02} whereas, for nondeterministic max-plus automata, this problem is undecidable \cite{Krob94}. 

We also consider an extended version of WAL which captures nondeterministic weighted B\"uchi automata. In extended WAL we allow existential quantification over first-order and second-order variables in the prefix of a formula. The structure of extended WAL is similar to the structure of unweighted logics for, e.g., timed automata \cite{Wil94} and data automata \cite{Bou02}. 

For the proof of our expressiveness equivalence result, we establish a Nivat decomposition theorem for nondeterministic and unambiguous weighted B\"uchi automata. Recall that Nivat's theorem \cite{Ni68} is one of the fundamental characterizations of rational transductions and shows a connection between rational transductions and rational language. Recently, Nivat's theorem was proved for semiring-weighted automata on finite words \cite{DK11} and weighted multioperator tree automata \cite{SVF09}. We obtain similar decompositions for WAL and extended WAL and deduce our results from the classical B\"uchi theorem \cite{Buc60}. Our proof is constructive and hence decidability properties for WAL and extended WAL can be transferred into decidability properties of weighted B\"uchi automata. 
As a side application of our Nivat theorem, we can easily show that weighted B\"uchi automata and weighted Muller automata are expressively equivalent.

\medskip

{\bf Outline.} In Sect. \ref{Sect:WBA} we introduce a general framework for weighted B\"uchi automata and consider several examples. In Sect. \ref{Sect:Decomp} we prove a Nivat decomposition theorem for weighted B\"uchi automata. In Sect. 4 we define weight assignment logic and its extension. In Sect. \ref{Sect:Exp} we state our main result and give a sketch of its proof for the unambiguous and nondeterministic cases. 

\section{Weighted B\"uchi Automata}
\label{Sect:WBA}

Let $\N = \{0, 1, ...\}$ denote the set of all natural numbers. For an arbitrary set $X$, an {\em $\omega$-word} over $X$ is an infinite sequence $(x_i)_{i \in \N}$ where $x_i \in X$ for all $i \in \N$. Let $X^{\omega}$ denote the set of all $\omega$-words over $X$. Any set $\mathcal L \subseteq X^{\omega}$ is called an {\em $\omega$-language} \index{$\omega$-language} over $X$.

A {\em B\"uchi automaton} over an alphabet $\Sigma$ is a tuple $\A = (Q, I, T, F)$ where $Q$ is a finite set of states, $\Sigma$ is an alphabet (i.e. a finite non-empty set), $I, F \subseteq Q$ are sets of initial resp. accepting states, and $T \subseteq Q \times \Sigma \times Q$ is a transition relation.  A {\em run} ${\rho = (t_i)_{i \in \N}} \in T^{\omega}$ of $\A$ is defined as an infinite sequence of matching transitions which starts in an initial state and visits some accepting state infinitely often, i.e., $t_i = (q_{i}, a_i, q_{i+1})$ for each $i \in \N$, such that $q_0 \in I$ and ${\{q \in Q \; | \; q = q_i \text{ for infinitely many } i \in \N\} \cap F \neq \emptyset}.$
Let ${\Label(\rho) := (a_i)_{i \in \N} \in \Sigma^{\omega}}$, the {\em label} of $\rho$. We denote by $\Run_{\A}$ \index{$\Run_{\A}$} the set of all runs of $\A$ and, for each $w \in \Sigma^{\omega}$, we denote by $\Run_{\A}(w)$\index{$\Run_{\A}(w)$}  the set of all runs $\rho$ of $\A$ with $\Label(\rho) = w$.
Let $\mathcal L(\A) = {\{w \in \Sigma^{\omega} \; | \; \Run_{\A}(w) \neq \emptyset\}}$,  the $\omega$-language {\em accepted} by $\A$. We call an $\omega$-language $\mathcal L \subseteq \Sigma^{\omega}$ {\em recognizable} \index{$\omega$-language!recognizable} if there exists a B\"uchi automaton $\A$ over $\Sigma$ such that $\mathcal L(\A) = \mathcal L$.

We say that a monoid $\mathbb K = (K, +, \zero)$ is {\em complete} \index{monoid!complete} (cf., e.g., \cite{DR06}) if it is equipped with infinitary sum operations $\sum_I: K^I \to K$ for any index set $I$, such that, for all $I$ and all families $(k_i)_{i \in I}$ of elements of $K$, the following hold:
\begin{itemize}
\item $\sum_{i \in \emptyset} k_i = \zero$, \; $\sum_{i \in \{j\}} k_i = k_j$, \; $\sum_{i \in \{p,q\}} k_i = k_p + k_q$ for $p \neq q$;
\item $\sum_{j \in J} (\sum_{i \in I_j} k_i) = \sum_{i \in I} k_i$, if $\bigcup_{j \in J} I_j = I$ and $I_j \cap I_{j'} = \emptyset$ for $j \neq j'$.
\end{itemize}
Let $\RR = \mathbb R \cup \{-\infty, \infty\}$. Then, $\RR$ equipped with infinitary operations like {\em infinum} or {\em supremum} forms a complete monoid. Now we introduce an algebraic structure for weighted B\"uchi automata which is an extension of totally complete
semirings \cite{DR06}  and valuation monoids \cite{DM12} and covers various multi-weighted measures.

\begin{definition}
A {\em valuation structure} $\mathbb V = (M, \mathbb K, \val)$ consists of a non-empty set $M$, a complete monoid $\mathbb K = (K, +, \zero)$ and a
mapping $\val: M^{\omega} \to K$ called henceforth a {\em valuation function}.
\end{definition}

In the definition of a valuation structure we have two weight domains $M$ and $K$. Here $M$ is the set of transition weights which in the multi-weighted examples can be tuples of weights (e.g., a reward-cost pair) and $K$ is the set of  weights of computations which can be single values (e.g., the ratio between rewards and costs).

\begin{definition}
Let $\Sigma$ be an alphabet and  $\mathbb V = (M, (K, +, \zero), \val)$ a valuation structure. A {\em weighted B\"uchi automaton} (WBA) over $\mathbb V$ is a tuple $\A = (Q, I, T, F, \wt)$ where $(Q, I, T, F)$ is a B\"uchi automaton over $\Sigma$ and $\wt: T \to M$ is a transition weight function.
\end{definition}

The behavior of WBA is defined as follows. Given a run $\rho$ of this automaton, we evaluate the $\omega$-sequence of transition weights of $\rho$ (which is in $M^{\omega}$) using the valuation function $\val$ and then resolve the nondeterminism on the weights of runs using the complete monoid $\mathbb K$.  Formally, let $\rho = (t_i)_{i \in \N} \in T^{\omega}$ be a run of $\A$. Then, the {\em weight} of $\rho$ is defined as $\wt_{\A}(\rho) = \val((\wt(t_i))_{i \in \mathbb N}) \in K$. \index{$\wt_{\A}(\rho)$}The {\em behavior} of $\A$ is a mapping ${[\![\A]\!]: \Sigma^{\omega} \to K}$ \index{$[{\hspace{-0.05cm}}[\A]{\hspace{-0.05cm}}]$} defined for all $w \in \Sigma^{\omega}$  by $[\![\A]\!](w) = \sum (\wt_{\A}(\rho) \; | \; \rho \in \Run_{\A}(w)).$
Note that the sum in the equation above can be infinite. Therefore we consider a complete monoid $(K, +, \zero)$.
A mapping $\mathbb L: \Sigma^{\omega} \to K$ is called a {\em quantitative $\omega$-language}. \index{quantitative $\omega$-language} We say that $\mathbb L$ is (nondeterministically) {\em recognizable} over $\mathbb V$ if there exists a WBA $\A$ over $\mathbb V$ such that $[\![\A]\!] = \mathbb L$. 

We say that a WBA $\A$ over $\Sigma$ and $\mathbb V$ is {\em unambiguous} if $|\Run_{\A}(w)| \le 1$ for every $w \in \Sigma^{\omega}$. We call a quantitative $\omega$-language $\mathbb L: \Sigma^{\omega} \to K$  {\em unambiguously recognizable} over $\mathbb V$ if there exists an unambiguous WBA $\A$ over $\Sigma$ and $\mathbb V$ such that $[\![\A]\!] = \mathbb L$. 

\begin{example}
\label{Example:OmegaValuationStructures}
\begin{itemize}
\item [(a)] The ratio measure was introduced in \cite{BBL08}, e.g., for the modeling of the average costs in timed systems. In the setting of $\omega$-words, we consider the model with two weight parameters: the {\em cost} and the {\em reward}. The rewards and costs of transitions are accumulated along every finite prefix of a run and their ratio is taken. Then, the weight of an infinite run is defined as the limit superior (or limit inferior) of the sequence of the computed ratios for all finite prefixes. To describe the behavior of these double-priced ratio B\"uchi automata, we consider the valuation structure $\mathbb V^{\Ratio} = (M, \mathbb K, \val)$ where $M = \mathbb Q \times \mathbb Q_{\ge 0}$ models the reward-cost pairs, $\mathbb K = (\RR, \sup, -\infty)$ and $\val: M^{\omega} \to \RR$ is defined for every sequence $u = ((r_i, c_i))_{i \in \mathbb N} \in M^{\omega}$ by $\val(u) = \limsup_{n \to \infty} \frac{r_1 + ... + r_n}{c_1 + ... + c_n}$. Here, we assume that $\frac{r}{0} = -\infty$.
\item [(b)] {\em Discounting} \cite{And06, CDH08} is a well-known principle which is used in, e.g., economics and psychology. In this example, we consider WBA with transition-dependent discounting, i.e., are two weight parameters: the cost and the discounting factor (which is not fixed and depends on a transition). In order to define WBA with discounting formally, we consider 
the valuation structure ${\mathbb V^{\Disc} = (M, \mathbb K, \val)}$ where ${M = \mathbb Q_{\ge 0} \times ((0, 1] \cap \mathbb Q)}$ models the pairs of a cost and a discounting factor, $\mathbb K = (\mathbb R_{\ge 0} \cup \{\infty\}, \inf, \infty)$, and $\val$ is defined for all ${u = ((c_i, d_i))_{i \in \mathbb N} \in M^{\omega}}$  as
$
\val(u) = c_0 + \sum_{i = 1}^{\infty} c_i \cdot \prod_{j = 0}^{i-1} d_j.
$
\item [(c)] Now we consider the valuation structure for the model of multi-weighted automata which correspond to one-player energy games with lower bound considered in \cite{FJLLS11}. Let $n \ge 1$ and $s_1, ..., s_n$ be energy storages. We start with empty storages and, after taking a transition of a B\"uchi automaton, the energy level of each storage $s_j$ (${1 \le j \le n}$) can be increased (if we regain energy) or decreased (if we consume energy). The goal is to keep the energy level of every energy storage not less than zero. Consider the sequence $u = (u_i)_{u \in \mathbb N}$ where, for all $i \in \mathbb N$, $u_i = (u_i^1, ..., u_i^n) \in \mathbb Z^n$ is the vector of the energy level changes for each storage. We say that $u$ is {\em correct} if $\sum_{k = 0}^i u_k^j \ge 0$ for all $i \in \N$ and $j \in \{1, ..., n\}$. For this situation we consider the valuation structure $\mathbb V^{\Energy} = (M, \mathbb K, \val)$ where $M = \mathbb Z^n$, $\mathbb K = (\{0,1\}, \vee, 0)$ and, for all $u \in M^{\omega}$, we let $\val(u) = 1$ if $u$ is correct and $\val(u) = 0$ otherwise.
\item [(d)] Since a valuation monoid  $(K, (K, +, \zero), \val)$ of Droste and Meinecke \cite{DM12} is a special case of valuation structures, all examples considered there also fit into our framework.
\qed
\end{itemize}
\end{example}

\section{Decomposition of WBA}
\label{Sect:Decomp}

In this section, we establish a Nivat decomposition theorem for WBA. We will need it for the proof of our main result. However, it also could be of independent interest.

Let $\Sigma$ be an alphabet and $\mathbb V = (M, (K, +, \zero), \val)$ a valuation structure.
For a (possibly different from $\Sigma$) alphabet $\Gamma$, we introduce the following operations.
Let $\Delta$ be an arbitrary non-empty set and $h: \Gamma \to \Delta$ a mapping called henceforth a {\em renaming}. For any $\omega$-word $u = (\gamma_i)_{i \in \N} \in \Gamma^{\omega}$, we let $h(u) = (h(\gamma_i))_{i \in \N} \in \Delta^{\omega}$. Now let $h: \Gamma \to \Sigma$ be a renaming and $\mathbb L: \Gamma^{\omega} \to K$ a quantitative $\omega$-language. We define the {\em renaming} $h(\mathbb L): \Sigma^{\omega} \to K$ for all $w \in \Sigma^{\omega}$ by $h(\mathbb L)(w) = \sum \big( \mathbb L(u) \; | \; u \in \Gamma^{\omega} \text{ and } h(u) = w \big)$. For a renaming $g: \Gamma \to M$, the {\em composition} $\val \circ g: \Gamma^{\omega} \to K$ is defined for all $u \in \Gamma^{\omega}$ by $(\val \circ g)(u) = \val(g(u))$. Given a quantitative $\omega$-language $\mathbb L: \Gamma^{\omega} \to K$ and an $\omega$-language $\mathcal L \subseteq \Gamma^{\omega}$, the intersection $\mathbb L \cap \mathcal L: \Gamma^{\omega} \to K$ is defined for all $u \in \mathcal L$ as $(\mathbb L \cap \mathcal L)(u) = \mathbb L(u)$ and for all $u \in \Gamma^{\omega} \setminus \mathcal L$ as $(\mathbb L \cap \mathcal L)(u) = \zero$. Given a renaming $h: \Gamma \to \Sigma$ , we say that an $\omega$-language $\mathcal L \subseteq \Gamma^{\omega}$ is {\em $h$-unambiguous} if for all $w \in \Sigma^{\omega}$ there exists at most one $u \in \mathcal L$ such that $h(u) = w$.

Our Nivat decomposition theorem for WBA is the following.

\begin{theorem}
\label{Theorem:Nivat}
Let $\Sigma$ be an alphabet, $\mathbb V = (M, (K, +, \zero), \val)$ a valuation structure, and $\mathbb L: \Sigma^{\omega} \to K$ a quantitative $\omega$-language. Then
\begin{itemize}
\item [(a)] $\mathbb L$ is unambiguously recognizable over $\mathbb V$ iff there exist an alphabet $\Gamma$, renamings $h: \Gamma \to \Sigma$ and $g: \Gamma \to M$, and a recognizable and $h$-unambiguous $\omega$-language $\mathcal L \subseteq \Gamma^{\omega}$ such that $\mathbb L = h((\val \circ g) \cap \mathcal L)$. 
\item [(b)] $\mathbb L$ is nondeterministically recognizable over $\mathbb V$ iff there exist an alphabet $\Gamma$, renamings $h: \Gamma \to \Sigma$ and $g: \Gamma \to M$, and a recognizable $\omega$-language $\mathcal L \subseteq \Gamma^{\omega}$ such that $\mathbb L = h((\val \circ g) \cap \mathcal L)$.
\end{itemize}
\end{theorem}

\subsection{Proof of Theorem \ref{Theorem:Nivat}}

We start with part (b) of Theorem \ref{Theorem:Nivat}.

\begin{lemma}
\label{Lemma:3.2}
Let $\A$ be a WBA over $\Sigma$ and $\mathbb V$.  Then there exist an alphabet $\Gamma$, renaming $h: \Gamma \to \Sigma$ and $g: \Gamma \to M$, and a recognizable $\omega$-language $\mathcal L \subseteq \Gamma^{\omega}$ such that $[\![\A]\!] = h((\val \circ g) \cap \mathcal L)$.
\end{lemma}

\begin{proof}
The idea is as in \cite{DK11} to take the set of transitions as the extended alphabet $\Gamma$. Then, $h$ maps every transition to its label and $g$ maps every transition to its weight. Then, if in the underlying unweighted B\"uchi automaton we label every transition with itself, then we obtain the B\"uchi automaton accepting $\mathcal L$. 

Formally, let $\A = (Q, I, T, F, \wt)$. We may assume w.l.o.g. that $T \neq \emptyset$.We let ${\Gamma = T}$ and $h: \Gamma \to \Sigma$ be defined for all $t = (p, a, q) \in T$ as $h(t) = a$, and let ${g: T \to M}$ be defined for all $t \in T$ as $g(t) = \wt(t)$. We also define $\mathcal L$ by ${\mathcal L = \{\rho = (t_i)_{i \in \mathbb N} \; | \; \rho \text{ is a run of } \A\}}$. 

First we show that $\mathcal L$ is recognizable. Indeed, consider the B\"uchi automaton ${\A' = (Q, I, T', F)}$ over $\Gamma$ where $T' = \{(p, (p, a, q), q) \; | \; (p, a, q) \in T\}$. Then $\mathcal L(\A') = \mathcal L$ and hence $\mathcal L$ is recognizable. Finally we show that $[\![\A]\!] = h((\val^{\omega} \circ g) \cap \mathcal L)$. Let $w \in \Sigma^{\omega}$. Then
$$
h((\val^{\omega} \circ g) \cap \mathcal L)(w) = \sum_{\begin{subarray}{c} u \in \mathcal L, \\ h(u) = w \end{subarray}}
\val^{\omega}(g(u)) = \sum_{\rho \in \Run_{\A}(w)} \wt_{\A}(\rho) = [\![\A]\!](w).
$$
\qed
\end{proof}

Now we turn to the implication $\Leftarrow$.

\begin{lemma}
\label{Lemma:3.3}
Let $\Gamma$ be an alphabet, $h: \Gamma \to \Sigma$ and $g: \Gamma \to M$ renamings, and 
$\mathcal L \subseteq \Gamma^{\omega}$ a recognizable $\omega$-language. Then, the quantitative $\omega$-language 
$h((\val \circ g) \cap \mathcal L)$ is recognizable over $\mathbb V$.
\end{lemma}

\begin{proof}
Since B\"uchi automata are not determinizable, the most challenging part in the proof is to show that recognizability of quantitative $\omega$-languages is stable under intersection with recognizable $\omega$-languages. Here we apply the result of \cite{CM00} that recognizable $\omega$-languages are recognizable by unambiguous B\"uchi automata. 

Let $\A$ be an unambiguous B\"uchi automaton over $\Gamma$ with $\mathcal L(\A) = \mathcal L$. If we associate with every transition $(p, \gamma, q)$ of $\A$ the weight $g(\gamma) \in M$, then we obtain the WBA $\mathcal B$ over $\Gamma$ and $\mathbb V$ with $[\![\mathcal B]\!] = (\val \circ g) \cap \mathcal L$. It remains to show that recognizable quantitative $\omega$-languages are closed under renaming. For this, we apply the construction of Droste and Vogler \cite{DV12}. Let $\mathcal B = (Q, I, T, F, \wt)$. Then we construct the WBA $\mathcal C = (Q', I', T', F', \wt')$ over $\Sigma$ and $\mathbb V$ defined as follows:
\begin{itemize}
\item $Q' = Q \times \Gamma$, $I' = I \times \{\gamma_0\}$ for some fixed $\gamma_0 \in \Gamma$, $F' = F \times \Gamma$;
\item $T'$ consists of all transitions $t = ((p, \gamma), a, (p', \gamma')) \in Q' \times \Sigma \times Q'$ such that $(p, \gamma', p') \in T$ and $h(\gamma') = a$. For such a transition $t$, we let $\wt'(t) = \wt(p, \gamma', p')$. 
\end{itemize}
Then $h([\![\mathcal B]\!]) = [\![\mathcal C]\!]$. Hence the quantitative $\omega$-language $h((\val \circ g) \cap \mathcal L)$ is recognizable over $\mathbb V$. 
\qed
\end{proof}

Then Theorem \ref{Theorem:Nivat}(b) follows immediately from Lemmas \ref{Lemma:3.2} and \ref{Lemma:3.3}.

The proof of Theorem \ref{Theorem:Nivat}(a) relies on the same constructions as the proof of Theorem \ref{Theorem:Nivat}(b). Note that in the proof of Lemma \ref{Lemma:3.2}, if $\A$ is unambiguous, the $\omega$-language $\mathcal L$ is $h$-unambiguous. Note also that the WBA $\mathcal B$ in the proof of Lemma \ref{Lemma:3.3} is unambiguous but, in general, $\mathcal C$ is not. However, $h$-unambiguity of $\mathcal L$ guarantees that $\mathcal C$ is unambiguous.

\subsection{Weighted Muller Automata}

As a first application of Theorem \ref{Theorem:Nivat} we show that WBA are expressively equivalent to {\em weighted Muller automata} which are defined as WBA with the difference that a set of accepting states $F \subseteq Q$ is replaced by a set $\mathcal F \subseteq 2^{Q}$ of sets of accepting states. Then, for an accepting run $\rho$, the set of all states, which are visited in $\rho$ infinitely often, must be in $\mathcal F$. 

\begin{theorem}
\label{Theorem:BuechiMuller}
Let $\Sigma$ be an alphabet, $\mathbb V = (M, (K, +, \zero), \val)$ a valuation structure and $\mathbb L: \Sigma^{\omega} \to K$ a quantitative $\omega$-language. Then $\mathbb L = [\![\A]\!]$ for some WBA $\A$ over $\Sigma$ and $\mathbb V$ iff $\mathbb L = [\![\A']\!]$ for some weighted Muller automaton $\A'$ over $\Sigma$ and $\mathbb V$.
\end{theorem}

Theorem \ref{Theorem:BuechiMuller} extends the result of \cite{DR06} for totally complete semirings. Whereas the proof of \cite{DR06} was given by direct non-trivial automata transformation, our proof is based on the fact that weighted Muller automata permit the same decomposition as stated in Theorem \ref{Theorem:Nivat} for WBA. The constructions for this case are much the same as the constructions of Theorem \ref{Theorem:Nivat}(b). We only have to replace $F \subseteq Q$ by $\mathcal F \subseteq 2^Q$ in the proofs and slightly modify the constructions of Lemma \ref{Lemma:3.3}.
\begin{itemize}
\item It is well known that Muller automata are determinizable. Then, for the construction of the weighted Muller automaton for $(\val \circ g) \cap \mathcal L$ we use the fact that Muller and B\"uchi automata are expressively equivalent and take a deterministic Muller automaton recognizing $\mathcal L$.
\item In the definition of $\mathcal C$ in the proof of Lemma \ref{Lemma:3.3} we replace $F'$ by the Muller acceptance condition $\mathcal F'$ which consists of all sets $\{(q_1, \gamma_1), ..., (q_k, \gamma_k)\} \subseteq Q'$ such that $\{q_1, ..., q_k\} \in \mathcal F$ (a similar idea was used in \cite{DM12}). 
\end{itemize}

\section{Weight Assignment Logic}

\subsection{Partial $\omega$-words}

Before we give a definition of the syntax and semantics of our new logic, we introduce some auxiliary notions about partial
$\omega$-words. Let $X$ be an arbitrary non-empty set. A {\em partial $\omega$-word} over $X$ is a partial mapping $u: \mathbb N \dasharrow X$, i.e., $u: U \to X$ for some $U \subseteq \mathbb N$. Let $\dom(u) = U$, the {\em domain} of $u$. We denote by $X^{\uparrow}$ the set of all partial $\omega$-words over $X$. Clearly, $X^{\omega} \subseteq X^{\uparrow}$. A {\em trivial $\omega$-word} $\top \in X^{\uparrow}$ is the partial $\omega$-word with $\dom(\top) = \emptyset$. 
For $u \in X^{\uparrow}$, $i \in \N$ and $x \in X$, the {\em update} $u[i/x] \in X^{\uparrow}$ is defined as $\dom(u[i/x]) = \dom(u) \cup \{i\}$, $u[i/x](i) = x$ and $u[i/x](i') = u(i')$ for all $i' \in \dom(u) \setminus \{i\}$.
Let $\theta = (u_j)_{j \in J}$ be an arbitrary family of partial $\omega$-words $u_j \in X^{\uparrow}$ where $J$ is an arbitrary index set. We say that $\theta$ is {\em compatible} if, for all $j, j' \in J$ and $i \in \dom(u_j) \cap \dom(u_{j'})$, we have ${u_j(i) = u_{j'}(i)}$. 
If $\theta$ is compatible, then we define the {\em merging} $u := (\bigsqcap_{j \in J} u_j) \in X^{\uparrow}$ as ${\dom(u) = \bigcup_{j \in J} \dom(u_j)}$ and, for all ${i \in \dom(u)}$, $u(i) = u_j(i)$ whenever $i \in \dom(u_j)$ for some $j \in J$. 
Let $\theta = \{u_j\}_{j \in \{1,2\}}$ be compatible. Then, we write $u_1 \uparrow u_2$.  Clearly, the relation $\uparrow$ is reflexive and symmetric. In the case $u_1 \uparrow u_2$, for $\bigsqcap_{j \in \{1,2\}} u_j$ we will also use notation $u_1 \sqcap u_2$.

\begin{example}
Let $X = \{a, b\}$ with $a \neq b$ and $u_1 = a^{\omega} \in X^{\uparrow}$. Let $u_2 \in X^{\uparrow}$ be the partial $\omega$-word whose domain $\dom(u_2)$ is the set of all odd natural numbers and $u_2(i) = a$ for all $i \in \dom(u_2)$. Let $u_3 \in X^{\uparrow}$ be the partial $\omega$-word such that $\dom(u_3)$ is the set of all even natural numbers and $u_3(i) = b$ for all $i \in \dom(u_3)$. Then $u_1 \uparrow u_2$ and $u_2 \uparrow u_3$,  but $\lnot (u_1 \uparrow u_3)$. This shows in particular that the relation $\uparrow$ is not transitive if $X$ is not a singleton set.  Then, $u_1 \sqcap u_2 = a^{\omega}$ and $u_2 \sqcap u_3 = (ba)^{\omega}$.
\end{example}

\subsection{WAL: Syntax and Semantics}

Let $V_1$ be a countable set of {\em first-order variables} and $V_2$ a countable set of {\em second-order variables}  such that $V_1 \cap V_2 = \emptyset$. Let $V = V_1 \cup V_2$.  Let $\Sigma$ be an alphabet and $\mathbb V = (M, (K, +, \zero), \val)$ a valuation structure. We also consider a designated element $\one \in M$ which we call the {\em default weight}. We denote the pair $(\mathbb V, \one)$ by $\mathbb V_{\one}$. The set $\uWAL(\Sigma, \mathbb V_{\one})$ of formulas of {\em weight assignment logic} over $\Sigma$ and $\mathbb V_{\one}$ is given by the grammar
$$
\varphi \; ::= \; P_a(x) \; | \; x = y \; | \; \; x < y \; | \; X(x) \; | \;  x \mapsto m \; | \; \varphi \Rightarrow \varphi \; | \;  \varphi \sqcap \varphi \; | \; {\sqcap} x. \varphi \; | \; {\sqcap} X. \varphi
$$
where $a \in \Sigma$, $x, y \in V_1$, $X \in V_2$ and $m \in M$. Such a formula $\varphi$ is called a {\em weight assignment formula}. 

Let $\varphi \in \uWAL(\Sigma, \mathbb V_{\one})$. We denote by $\Const(\varphi) \subseteq M$ the set of all weights $m \in M$ occurring in $\varphi$. The set $\Free(\varphi) \subseteq V$ of {\em free variables} of $\varphi$ is defined to be the set of all variables $\mathcal X \in V$ which appear in $\varphi$ and are not bound by any quantifier ${\sqcap} \mathcal X$.
We say that $\varphi$ is a {\em sentence} if $\Free(\varphi) = \emptyset$.

Note that the merging as defined before is a partially defined operation, i.e., it is defined only for compatible families of partial $\omega$-words. In order to extend it to a totally defined operation, we fix an element $\bot \notin M^{\uparrow}$ which will mean the undefined value. Let $M^{\uparrow}_{\bot} = M^{\uparrow} \cup \{\bot\}$. Then, for any family $\theta = (u_j)_{j \in J}$ with $u_j \in M^{\uparrow}_{\bot}$, such that either $\theta \in (M^{\uparrow})^J$ is not compatible or $\theta \in (M_{\bot}^{\uparrow})^J \setminus (M^{\uparrow})^J$, we let $\bigsqcap_{j \in J} u_j = \bot$.

For any $\omega$-word $w \in \Sigma^{\omega}$, a {\em $w$-assignment} is a mapping $\sigma: V \to \dom(w) \cup 2^{\dom(w)}$ mapping first-order variables to elements in $\dom(w)$ and second-order variables to subsets of $\dom(w)$.
For a first-order variable $x$ and a position $i \in \mathbb N$, the $w$-assignment $\sigma[x/i]$ is defined on $V \setminus \{x\}$ as $\sigma$, and we let $\sigma[x/i](x) = i$. For a second-order variable $X$ and a subset $I \subseteq \mathbb N$, the $w$-assignment $\sigma[X/I]$ is defined similarly.  Let $\Sigma^{\omega}_V$ denote the set of all pairs $(w, \sigma)$ where $w \in \Sigma^{\omega}$ and $\sigma$ is a $w$-assignment. We will denote such pairs $(w, \sigma)$ by $w_{\sigma}$. 

The semantics of $\uWAL$-formulas is defined in two steps: by means of the auxiliary and proper semantics. 
Let $\varphi \in \uWAL(\Sigma, \mathbb V_{\one})$. The {\em auxiliary semantics} of $\varphi$ is the mapping $\llangle \varphi \rrangle: \Sigma_V^{\omega} \to M^{\uparrow}_{\bot}$ defined for all $w_{\sigma} \in \Sigma_V^{\omega}$ with $w = (a_i)_{i \in \N}$ as shown in Table \ref{Table:Auxiliary_Semantics}. Note that the definition of $\llangle .. \rrangle$ does not employ $+$ and $\val$. The {\em proper semantics} ${[\![\varphi]\!]: \Sigma^{\omega}_V \to K}$ operates on the auxiliary semantics $\llangle \varphi \rrangle$ as follows. Let $w_{\sigma} \in \Sigma_V^{\omega}$. If $\llangle \varphi \rrangle (w_{\sigma}) \in M^{\uparrow} $, then we assign the default weight to all undefined positions in $\dom(\llangle \varphi \rrangle (w_{\sigma}))$ and evaluate the obtained sequence using $\val$. Otherwise, if $\llangle \varphi \rrangle (w_{\sigma}) = \bot$, we put $[\![\varphi]\!](w_{\sigma}) = \zero$. Note that if $\varphi \in \uWAL(\Sigma, \mathbb V_{\one})$ is a sentence, then the values $\llangle \varphi \rrangle (w_{\sigma})$ and $[\![\varphi]\!](w_{\sigma})$ do not depend on $\sigma$ and we consider the auxiliary semantics of $\varphi$ as 
the mapping $\llangle \varphi \rrangle : \Sigma^{\omega} \to M^{\uparrow}_{\bot}$ and the proper semantics of $\varphi$ as the quantitative $\omega$-language $[\![\varphi]\!]: \Sigma^{\omega} \to K$. 
Note that $+$ was not needed for the semantics of $\uWAL$-formulas. This operation will be needed in the next section for the extension of $\uWAL$. We say that a quantitative $\omega$-language $\mathbb L: \Sigma^{\omega} \to K$ is {\em $\uWAL$-definable} over $\mathbb V$ if there exist a default weight $\one \in M$ and a sentence ${\varphi \in \uWAL(\Sigma, \mathbb V_{\one})}$ such that $[\![\varphi]\!] = \mathbb L$.

\begin{table}[t]
{\footnotesize
\begin{center}
\begin{minipage}{0.49\linewidth}
$\begin{array}{@{}l@{\hspace{0.1cm}}l@{\hspace{0.1cm}}l}
\llangle P_a(x) \rrangle (w_{\sigma}) & = & \begin{cases}
\top, & a_{\sigma(x)} = a \\
\bot, & \text{otherwise}
\end{cases} \\
\llangle x = y \rrangle (w_{\sigma}) &= & \begin{cases}
\top, & \sigma(x) = \sigma(y) \\
\bot, & \text{otherwise}
\end{cases} \\
\llangle x < y \rrangle (w_{\sigma}) &= & \begin{cases}
\top, & \sigma(x) < \sigma(y) \\
\bot, & \text{otherwise}
\end{cases} \\
\llangle X(x) \rrangle (w_{\sigma}) & = & \begin{cases}
\top, &  \sigma(x) \in \sigma(X) \\
\bot, & \text{otherwise}
\end{cases} 
\end{array}$ 
\end{minipage}
\hfill
\begin{minipage}{0.49\linewidth}
$\begin{array}{@{\hspace{-0.7cm}}l@{\hspace{0.1cm}}l@{\hspace{0.1cm}}l}
\llangle x \mapsto m \rrangle (w_{\sigma}) & = & \top[\sigma(x)/m] \\
\llangle \varphi_1 \Rightarrow \varphi_2 \rrangle (w_{\sigma}) & = & \begin{cases}
\llangle \varphi_2 \rrangle (w_{\sigma}), & \llangle \varphi_1 \rrangle (w_{\sigma}) = \top \\
\top, & \text{otherwise} \end{cases} \\
\llangle \varphi_1 \sqcap \varphi_2 \rrangle (w_{\sigma}) & = & \llangle \varphi_1 \rrangle (w_{\sigma}) \sqcap \llangle \varphi_2 \rrangle (w_{\sigma}) \\
\llangle {\sqcap} x. \varphi \rrangle (w_{\sigma}) & = & \bigsqcap_{i \in \dom(w)} \llangle \varphi \rrangle (w_{\sigma[x/i]}) \\
\llangle {\sqcap} X. \varphi \rrangle (w_{\sigma}) & = & \bigsqcap_{I \subseteq \dom(w)} \llangle \varphi \rrangle (w_{\sigma[X/I]})
\end{array}$ 
\end{minipage}
\end{center}
}
\caption{The auxiliary semantics of $\uWAL$-formulas}
\label{Table:Auxiliary_Semantics}
\end{table}

\begin{example} 
Consider a valuation structure $\mathbb V = (M, (K, +, \zero), \val)$ and a default weight $\one \in M$. Consider an alphabet $\Sigma = \{a, b, ...\}$ of actions. We assume that the cost of $a$ is $c(a) \in M$, the cost of $b$ is $c(b) \in M$, and the costs of all other actions $x$ in $\Sigma$ are equal to $c(x) = \one$ (which can mean, e.g., that these actions do not invoke any costs). Then every $\omega$-word $w$ induces the $\omega$-word of costs. We want to construct a sentence of our WAL which for every such an $\omega$-word will evaluate its sequence of costs using $\val$. The desired sentence $\varphi \in \uWAL(\Sigma, \mathbb V_{\one})$ is
$$
{\varphi = {\sqcap} x. ([P_a(x) \Rightarrow (x \mapsto c(a))] \sqcap [P_b(x) \Rightarrow (x \mapsto c(b))]).}
$$
Then, for every $w = (a_i)_{i \in \mathbb N} \in \Sigma^{\omega}$, the auxiliary semantics $\llangle \varphi \rrangle(w)$ is the partial $\omega$-word over $M$ where all positions $i \in \N$ with $a_i = a$ are labelled by $c(a)$, all positions with $a_i = b$ are labelled by $c(b)$, and the labels of all other positions are undefined. Then, the proper semantics $[\![\varphi]\!](w)$ assigns $\one$ to all positions with undefined labels and evaluates it by means of $\val$.
\end{example}

\subsection{WAL: Relation to MSO Logic}

Let $\Sigma$ be an alphabet. We consider monadic second-order logic $\MSO(\Sigma)$ over $\omega$-words to be the set of formulas 
$$
\varphi \; ::= \; P_a(x) \; | \; x = y  \; | \; x < y \; | \; X(x) \; | \; \varphi \wedge \varphi \; | \; \lnot \varphi \; | \; \forall x. \varphi \; | \; \forall X. \varphi
$$
where $a \in \Sigma$, $x,y \in V_1$ and $X \in V_2$. For $w_{\sigma} \in \Sigma^{\omega}_V$, the satisfaction relation $w_{\sigma} \models \varphi$ is defined as usual. The usual formulas of the form $\varphi_1 \vee \varphi_2$, $\exists \mathcal X. \varphi$ with $\mathcal X \in V$, $\varphi_1 \Rightarrow \varphi_2$ and $\varphi_1 \Leftrightarrow \varphi_2$ can be expressed using $\MSO$-formulas.

For any formula $\varphi \in \MSO(\Sigma)$, let $W(\varphi)$ denote the $\uWAL$-formula obtained from $\varphi$ by replacing $\wedge$ by $\sqcap$, $\forall \mathcal X$ (with $\mathcal X \in V$) by $\sqcap \mathcal X$, and every subformula $\lnot \psi$ by ${\psi \Rightarrow \false}$. Here $\false$ can be considered as abbreviation of the sentence ${\sqcap} x. (x < x)$.
Note that $W(\varphi)$ does not contain any assignment formulas $x \mapsto m$ and $\llangle W(\varphi) \rrangle (w_{\sigma}) \in \{\top, \bot\}$ for every $w_{\sigma} \in \Sigma_V^{\omega}$. Moreover, it can be easily shown by induction on the structure of $\varphi$ that, for all $w_{\sigma} \in \Sigma_V^{\omega}$: $w_{\sigma} \models \varphi$ iff $\llangle W(\varphi) \rrangle (w_{\sigma}) = \top$. This shows that MSO logic on infinite words is subsumed by $\uWAL$. For the formulas which do not contain any assignments of the form $x \mapsto m$, the merging $\sqcap$ can be considered as the usual conjunction and the merging quantifiers ${\sqcap} \mathcal X$ as the usual universal quantifiers $\forall \mathcal X$. Moreover, $\top$ corresponds to the boolean true value and $\bot$ to the boolean false value.

For a $\uWAL$-formula $\varphi$, we will consider $\lnot \varphi$ as abbreviation for $\varphi \Rightarrow \false$.

\subsection{Extended WAL}

Here we extend $\uWAL$ with weighted existential quantification over free variables in $\uWAL$-formulas.
Let $\Sigma$ be an alphabet, $\mathbb V = (M, (K, +, \zero), \val)$ a valuation structure and $\one \in M$ a default weight. The set $\eWAL(\Sigma, \mathbb V_{\one})$ of formulas of {\em extended weight assignment logic} over $\Sigma$ and $\mathbb V_{\one}$ consists of all formulas of the form ${\sqcup} \mathcal X_1. \; ... \; {\sqcup} \mathcal X_k. \varphi$ where $k \ge 0$, $\mathcal X_1, ..., \mathcal X_k \in V$ and $\varphi \in \uWAL(\Sigma, \mathbb V_{\one})$.
Given a formula $\varphi \in \eWAL(\Sigma, \mathbb V_{\one})$, the {\em semantics} of $\varphi$ is the mapping $[\![\varphi]\!]: \Sigma^{\omega}_V \to K$ defined inductively as follows. If $\varphi \in \uWAL(\Sigma, \mathbb V_{\one})$, then $[\![\varphi]\!]$ is defined as the proper semantics for $\uWAL$. If $\varphi$ contains a prefix ${\sqcup} x$ with $x \in V_1$ or ${\sqcup} X$ with $X \in V_2$, then, for all $w_{\sigma} \in \Sigma_V^{\omega}$, $[\![\varphi]\!](w_{\sigma})$ is defined inductively as shown in Table \ref{Table:eWAL_Semantics}. Again, if $\varphi$ is a sentence, then we can consider its semantics as the quantitative $\omega$-language $[\![\varphi]\!]: \Sigma^{\omega} \to K$. We say that a quantitative $\omega$-language $\mathbb L: \Sigma^{\omega} \to K$ is {\em $\eWAL$-recognizable} over $\mathbb V$ if there exist a default weight $\one \in M$ and a sentence $\varphi \in \eWAL(\Sigma, \mathbb V_{\one})$ such that $[\![\varphi]\!] = \mathbb L$.

\begin{table}[t]
$$\begin{array}{@{}l@{\hspace{0.1cm}}l@{\hspace{0.1cm}}l}
\,[\![\sqcup x. \varphi]\!](w_{\sigma}) &=& \sum\big( [\![\varphi]\!](w_{\sigma[x/i]}) \; | \; i \in \dom(w) \big) \\
\,[\![\sqcup X. \varphi]\!](w_{\sigma}) &=& \sum\big( [\![\varphi]\!](w_{\sigma[X/I]}) \; | \; I \subseteq \dom(w) \big)
\end{array}$$
\caption{The semantics of $\eWAL$-formulas}
\label{Table:eWAL_Semantics}
\end{table}

\begin{example}
Let $\Sigma = \{a\}$ be a singleton alphabet, $\mathbb V = \mathbb V^{\Disc}$ as defined in Example \ref{Example:OmegaValuationStructures}(b). Assume that, for every position of an $\omega$-word, we can either assign to this position the cost $5$ and the discounting factor $0.5$ or we assign the cost the smaller cost $2$ and the bigger discounting factor $0.75$. After that we compute the discounted sum using the valuation function of $\mathbb V^{\Disc}$. We are interested in the infimal value of this discounted sum. We can express it by means of the $\eWAL$-formula
$$
{\varphi = {\sqcup X}. {\sqcap x}. ([X(x) \Rightarrow (x \mapsto (5, 0.5))] \sqcap [(\lnot X(x)) \Rightarrow (x \mapsto (2, 0.75))])}
$$
i.e. $[\![\varphi]\!](a^{\omega})$ is the desired infimal value.
\end{example}

\section{Expressiveness Equivalence Result}
\label{Sect:Exp}

In this section we state and prove the main result of this paper. 

\begin{theorem}
\label{Theorem:Main}
Let $\Sigma$ be an alphabet, $\mathbb V = (M, (K, +, \zero), \val)$ a valuation structure and $\mathbb L: \Sigma^{\omega} \to K$ a quantitative $\omega$-language. Then
\begin{itemize}
\item [(a)] $\mathbb L$ is $\uWAL$-definable over $\mathbb V$ iff $\mathbb L$ is unambiguously recognizable over $\mathbb V$.
\item [(b)] $\mathbb L$ is $\eWAL$-definable over $\mathbb V$ iff $\mathbb L$ is recognizable over $\mathbb V$.
\end{itemize}
\end{theorem}

\subsection{Unambiguous Case: Definability Implies Recognizability}
\label{Subsect:Unamb_Case}

In this subsection, we prove part (a) of Theorem \ref{Theorem:Main}. First we show $\uWAL$-definability implies unambiguous recognizability.  We establish a decomposition of $\uWAL$-formulas in a similar manner as it was done for unambiguous WBA in Theorem  \ref{Theorem:Nivat} (a), i.e., we separate weighted part of $\uWAL$ from its unweighted part. Then applying the classical B\"uchi theorem and our Nivat Theorem \ref{Theorem:Nivat}(a), we obtain that $\mathbb L$ is recognizable over $\mathbb V$. 

\begin{lemma}
\label{Lemma:UnDefRec}
Let $\varphi \in \uWAL(\Sigma, \mathbb V_{\one})$ be a sentence. Then there exist an alphabet $\Gamma$, renamings $h: \Gamma \to \Sigma$ and $g: \Gamma \to M$, and a sentence $\beta \in \MSO(\Gamma)$ such that 
$[\![\varphi]\!] = h((\val \circ g) \cap \mathcal L(\beta))$.
\end{lemma}

The proof of this lemma will be given in the rest of this subsection.
Let $\# \notin M$ be a symbol which we will use to mark all positions whose labels are undefined in the auxiliary semantics of $\uWAL$-formulas. Let $\Delta_{\varphi} = \Const(\varphi) \cup \{\#\}$. Then our extended alphabet will be $\Gamma = \Sigma \times \Delta_{\varphi}$. We define the renamings $h, g$ as follows.
For all $u = (a, b) \in \Gamma$, we let $h(u) = a$, $g(u) = b$ if $b \in M$, and $g(u) = \one$ if $m = \#$.
The main difficulty is to construct the sentence $\beta$. For any $\omega$-word $w = (a_i)_{i \in \mathbb N} \in \Sigma^{\omega}$ and any partial $\omega$-word $\eta \in (\Const(\varphi))^{\uparrow}$, we encode the pair $(w, \eta)$ as the $\omega$-word $\code(w, \eta) = ((a_i, b_i))_{i \in \N} \in \Gamma^{\omega}$ where, for all $i \in \dom(\eta)$, $b_i = \eta(i)$ and, for all $i \in \N \setminus \dom(\eta)$, $b_i = \#$. In other words, we will consider $\omega$-words of $\Gamma$ as convolutions of $\omega$-words over $\Sigma$ with the encoding of the auxiliary semantics of $\varphi$.

The construction of $\beta$ is based on the following technical lemma.

\begin{lemma}
\label{Lemma:Aux}
For every subformula $\zeta$ of $\varphi$, there exists a formula ${\Phi(\zeta) \in \MSO(\Sigma \times \Delta_{\varphi})}$ such that $\Free(\Phi(\zeta)) = \Free(\zeta)$ and, for all $w_{\sigma} \in \Sigma^{\omega}_V$ and $\eta \in (\Const(\varphi))^{\uparrow}$, we have: $\llangle \zeta \rrangle(w_{\sigma}) = \eta$ iff $(\code(w, \eta))_{\sigma} \models \Phi(\zeta)$.
\end{lemma}

Note that $\llangle \varphi \rrangle (w_{\sigma}) = \eta$ means in particular that $\llangle \varphi \rrangle (w_{\sigma}) \neq \bot$.

\begin{proof}
Let $Y \in V_2$ be a fresh variable which does not occur in $\varphi$.  First, we define inductively the formula ${\Phi_Y(\zeta) \in \MSO(\Gamma)}$ with $\Free(\Phi_{Y}(\zeta)) = \Free(\zeta) \cup \{Y\}$ which describes the connection between the input $\omega$-word $w$ and the output partial $\omega$-word $\eta$; here the variable $Y$ keeps track of the domain of $\eta$.
\begin{itemize}
\item For $\zeta = P_a(x)$, we let 
$$\Phi_Y(\zeta) = \bigvee_{b \in \Delta_{\varphi}} P_{(a,b)}(x) \wedge Y(\emptyset)$$ where $Y(\emptyset)$ is abbreviation for $\forall y. \lnot Y(y)$. Here we demand that the first component of the letter at position $x$ is $a$ and the second component is an arbitrary letter from $\Delta_{\varphi}$ and that the auxiliary semantics of $\zeta$ is the trivial partial $\omega$-word $\top$.
\item Let $\zeta$ be one of the formulas of the form $x = y$, $x < y$ or $X(x)$. Then, we let $$\Phi_Y(\zeta) =  \zeta \wedge Y(\emptyset).$$
\item For $\zeta = (x \mapsto m)$, we let $$\Phi_{Y}(\zeta) = \bigvee_{a \in \Sigma} P_{(a,m)}(x) \wedge \forall y. (Y(y) \Leftrightarrow x = y).$$ This formula describes that position $x$ of $\eta$ must be labelled by $m$ and all other positions are unlabelled. 
\item Let $\zeta = (\zeta_1 \Rightarrow \zeta_2)$. Let $Z \in V_2$ be a fresh variable. Consider the formula ${\kappa = \exists Z. [\Phi_Z(\zeta_1) \wedge Z(\emptyset)]}$ which checks whether the value of the auxiliary semantics of $\zeta_1$ is $\top$. Then, we let 
$$
\Phi_Y(\zeta) = (\kappa \wedge \Phi_Y(\zeta_2)) \vee (\lnot \kappa \wedge Y(\emptyset)).
$$
\item Let $\zeta = \zeta_1 \sqcap \zeta_2$. Let $Y_1, Y_2 \in V_2$ be two fresh distinct variables. Then we let $$\Phi_Y(\zeta) = \exists Y_1. \exists Y_2. (\Phi_{Y_1}(\zeta_1) \wedge \Phi_{Y_2}(\zeta_2) \wedge [Y = Y_1 \cup Y_2]).$$ Here $Y = Y_1 \cup Y_2$ is considered as abbreviation for the MSO-formula ${\forall y. (Y(y) \Leftrightarrow [Y_1(y) \vee Y_2(y)])}$.
\item The most interesting case is a formula of the form $\zeta = {\sqcap} \mathcal X. \zeta'$ with $\mathcal X \in V$. Here, every value of $\mathcal X$ induces its own value of $Y(\mathcal X)$ and we have to merge infinitely many partial $\omega$-words, i.e., to express that $Y$ is the infinite union of $Y(\mathcal X)$ over all sets $\mathcal X$. We can show that $Y$ must be the minimal set which satisfies the formula $\xi(Y) = \forall \mathcal X. \exists Y'. (\Phi_{Y'}(\zeta') \wedge (Y' \subseteq Y))$ where $Y' \in V_2$ is a fresh variable. Then, we let 
$$\Phi_Y(\zeta) = \xi(Y) \wedge \forall Z. (\xi(Z) \Rightarrow (Y \subseteq Z))$$ where $Z \in V_2$ is a fresh variable.
\end{itemize}

Let $w = (a_i)_{i \in \mathbb N} \in \Sigma^{\omega}$, $\sigma$ be a $w$-assignment and $\eta \in (\Const(\varphi))^{\uparrow}$. For $R \subseteq \mathbb N$, let ${\eta|_{R}  \in (\Const(\varphi))^{\uparrow}}$ be defined as $\dom(\eta|_R) = R \cap \dom(\eta)$ and $\eta|_R(i) = \eta(i)$ for all $i \in \dom(\eta|_R)$. Now we show by induction on the structure of $\zeta$ that
\begin{equation}
\label{Eq:Techno}
(\code(w, \eta))_{\sigma} \models \Phi_Y(\zeta) \; \text{ iff } \; \sigma(Y) \subseteq \dom(\eta) \text{ and } \llangle \zeta \rrangle (w_{\sigma}) = \eta|_{\sigma(Y)}.
\end{equation}

\begin{itemize}
\item Let $\zeta = P_a(x)$.
\begin{itemize}
\item Assume that $(\code(w, \eta))_{\sigma} \models \Phi_Y(\zeta)$. Then $a_{\sigma(x)} = a$ and $\sigma(Y) = \emptyset$. Hence $\llangle \zeta \rrangle(w_{\sigma}) = \top = \eta|_{\emptyset}$ and $\emptyset = \sigma(Y) \subseteq \dom(\eta)$.
\item Conversely, assume that $\sigma(Y) \subseteq \dom(\eta)$ and $\llangle \zeta \rrangle (w_{\sigma}) = \eta|_{\sigma(Y)}$. Then $\llangle \zeta \rrangle (w_{\sigma}) = \top$ which implies $a_{\sigma(x)} = a$ and $\sigma(Y) = \emptyset$. Then $(\code(w, \eta))_{\sigma} \models \Phi_Y(\zeta)$.
\end{itemize}
\item Let $\zeta$ be one of the formulas $x < y$, $x = y$ and $X(x)$.
\begin{itemize}
\item Assume that $(\code(w, \eta))_{\sigma} \models \Phi_Y(\zeta)$. Then $(\code(w, \eta))_{\sigma} \models \zeta$ and ${\sigma(Y) = \emptyset}$. Since $(\code(w, \eta))_{\sigma} \models \zeta$ implies $w_{\sigma} \models \zeta$, we obtain $\llangle \zeta \rrangle (w_{\sigma}) = \top = \eta|_{\emptyset}$ and $\emptyset = \sigma(Y) \subseteq \dom(\eta)$.
\item Conversely, assume that $\sigma(Y) \subseteq \dom(\eta)$ and $\llangle \zeta \rrangle (w_{\sigma}) = \eta|_{\sigma(Y)}$. Then $\llangle \zeta \rrangle (w_{\sigma}) = \top$ which implies $w_{\sigma} \models \zeta$ and $\sigma(Y) = \emptyset$.  Then, ${(\code(w, \eta))_{\sigma} \models \zeta}$ and $\sigma(Y) = \emptyset$. Hence $(\code(w, \eta))_{\sigma} \models \Phi_Y(\zeta)$. 
\end{itemize}
\item Let $\zeta = (x \mapsto m)$ with $m \in \Const(\varphi)$ (since $\zeta$ is a subformula of $\varphi$).
\begin{itemize}
\item Assume that $(\code(w, \eta))_{\sigma} \models \Phi_Y(\zeta)$. Then $\sigma(x) \in \dom(\eta)$, $\eta(\sigma(x)) = m$ and $\sigma(Y) = \{\sigma(x)\}$. Hence $\llangle \psi \rrangle (w_{\sigma}) = \top[\sigma(x)/m] = \eta|_{\sigma(Y)}$ and $\{\sigma(x)\} = Y \subseteq \dom(\eta)$.
\item Conversely, assume that the right hand side of (\ref{Eq:Techno}) holds true. Then $\eta|_{\sigma(Y)} = \top[\sigma(x)/m]$. Since $\sigma(Y) \subseteq \dom(\eta)$, we have $\sigma(Y) = \{\sigma(x)\}$. Moreover, $\eta(\sigma(x)) = m$. Then the left hand side of (\ref{Eq:Techno}) also holds true.
\end{itemize}
\item Let $\zeta = (\zeta_1 \Rightarrow \zeta_2)$.
\begin{itemize}
\item Assume that the left hand side of (\ref{Eq:Techno}) holds true. Then one of the following cases is possible.
\begin{itemize}
\item $(\code(w, \eta))_{\sigma} \models \kappa \wedge \Phi_Y(\zeta_2)$. Then, $(\code(w, \eta))_{\sigma[Y/\emptyset]} \models \Phi_Y(\zeta_1)$ and $(\code(w, \eta))_{\sigma} \models \Phi_Y(\zeta_2)$. Then by induction hypothesis for $\zeta_1$ and $\zeta_2$ we have: $\llangle \zeta_1 \rrangle (w_{\sigma}) = \eta|_{\emptyset} = \top$, $\sigma(Y) \subseteq \dom(\eta)$ and $\llangle \zeta_2 \rrangle (w_{\sigma}) = \eta|_{\sigma(Y)}$. This implies $\sigma(Y) \subseteq \dom(\eta)$ and $\llangle \zeta \rrangle (w_{\sigma}) = \llangle \zeta_2 \rrangle (w_{\sigma}) = \eta|_{\sigma(Y)}$. Hence the right hand side of (\ref{Eq:Techno}) holds true. 
\item $(\code(w, \eta))_{\sigma} \models \lnot \kappa \wedge Y(\emptyset)$. Then, $(\code(w, \eta))_{\sigma[Y/\emptyset]} \nvDash \Phi_Y(\zeta_1)$ and $\sigma(Y) = \emptyset$. Then by induction hypothesis for $\zeta_1$ we have $\llangle \eta_1 \rrangle (w_{\sigma}) \neq \eta|_{\emptyset} = \top$. Then $\emptyset = \sigma(Y) \subseteq \dom(\eta)$ and $\llangle \zeta \rrangle = \top = \eta|_{\sigma(Y)}$. Then the right hand side of (\ref{Eq:Techno}) holds true. 
\end{itemize}
\item Now assume that the right hand side of (\ref{Eq:Techno}) holds true. Then one of the following cases is possible.
\begin{itemize}
\item $\llangle \zeta_1 \rrangle (w_{\sigma}) = \top = \eta|_{\emptyset}$. Then by induction hypothesis for $\zeta_1$ we have $(\code(w, \eta))_{\sigma[Y/\emptyset]} \models \Phi_Y(\zeta_1)$ and hence $(\code(w, \eta))_{\sigma} \models \kappa$. Moreover, $\eta|_{\sigma(Y)} = \llangle \zeta \rrangle (w_{\sigma}) = \llangle \zeta_2 \rrangle (w_{\sigma})$ and $\sigma(Y) \subseteq \dom(\eta)$. Then by induction hypothesis for $\zeta_2$ we obtain $(\code(w, \eta))_{\sigma} \models \Phi_Y(\zeta_1)$. Then we have $(\code(w, \eta))_{\sigma} \models \kappa \wedge \Phi_Y(\zeta_1)$ and hence $(\code(w, \eta))_{\sigma} \models \Phi_Y(\zeta)$.
\item $\llangle \zeta_1 \rrangle (w_{\sigma}) \neq \top = \eta|_{\emptyset}$. Then by induction hypothesis for $\zeta_1$ we have $(\code(w, \eta))_{\sigma[Y/\emptyset]} \nvDash \Phi_Y(\zeta_1)$ and hence $(\code(w, \eta))_{\sigma} \nvDash \kappa$. Moreover, $\eta|_{\sigma(Y)} = \llangle \zeta \rrangle (w_{\sigma}) = \top = \eta|_{\emptyset}$ and $\sigma(Y) \subseteq \dom(\eta)$ which implies $\sigma(Y) = \emptyset$. Then $(\code(w, \eta))_{\sigma} \models \lnot \kappa \wedge Y(\emptyset)$ and hence $(\code(w, \eta))_{\sigma} \models \Phi_Y(\zeta)$.
\end{itemize}
\end{itemize}
\item Let $\zeta = \zeta_1 \sqcap \zeta_2$. 
\begin{itemize}
\item Assume that the left hand side of (\ref{Eq:Techno}) holds. Then there exist subsets $R_1, R_2 \subseteq \dom(w)$ such that:
\begin{itemize}
\item $\sigma(Y) = R_1 \cup R_2$,
\item $(\code(w, \eta))_{\sigma[Y/R_1]} \models \Phi_{Y}(\zeta_1)$,
\item $(\code(w, \eta))_{\sigma[Y/R_2]} \models \Phi_{Y}(\zeta_2)$.
\end{itemize}
Then by induction hypothesis for $\zeta_1$ and $\zeta_2$ we have:
\begin{itemize}
\item $R_1 \subseteq \dom(\eta)$ and $\llangle \zeta_1 \rrangle (w_{\sigma}) = \eta|_{R_1}$,
\item $R_2 \subseteq \dom(\eta)$ and $\llangle \zeta_2 \rrangle (w_{\sigma}) = \eta|_{R_2}$.
\end{itemize}
Then $\sigma(Y) \subseteq \dom(\eta)$ and, since $\eta|_{R_1}$ and $\eta|_{R_2}$ are compatible partial $\omega$-words, we have $\llangle \zeta_1 \sqcap \zeta_2 \rrangle (w_{\sigma} = \eta|_{R_1} \sqcap \eta|_{R_2} = \eta|_{\sigma(Y)}$. This shows that the right hand side of (\ref{Eq:Techno}) also holds true.
\item Conversely, assume that the right hand side of (\ref{Eq:Techno}) holds. Let $\eta_1 = \llangle \zeta_1 \rrangle (w_{\sigma})$ and $\eta_2 = \llangle \zeta_2 \rrangle (w_{\sigma})$. Then $\eta|_{\sigma(Y)} = \sigma_1 \sqcap \sigma_2$. Moreover, there exist $R_1, R_2 \subseteq \dom(w)$ such that:
\begin{itemize}
\item $R_1 \cup R_2 = \sigma(Y)$,
\item $\eta_1 = \eta|_{R_1}$ and $\eta = \eta|_{R_2}$.
\end{itemize}
Since $R_1, R_2 \subseteq \sigma(Y) \subseteq \dom(\eta)$, by induction hypothesis we have $(\code(w, \eta))_{\sigma[Y_i/R_i]} \models \Phi_Y(\zeta_i)$ for $i \in \{1,2\}$. Since $Y_2$ does not occur in $\Phi_Y(\zeta_1)$ and $Y_1$ does not occur in $\Phi_Y(\zeta_2)$, we have $\code(w, \eta)_{\sigma[Y_1/R_1][Y_2/R_2]} \models \Phi_Y(\zeta_i)$ for $i \in \{1,2\}$. Then the left hand side of (\ref{Eq:Techno}) holds. 
\end{itemize}
\item Let $\zeta = {\sqcap} x. \zeta'$ with $x \in V_1$.
\begin{itemize}
\item Assume that $(\code(w, \eta))_{\sigma} \models \Phi_Y(\zeta)$. Then $(\code(w, \eta))_{\sigma} \models \xi(Y)$. This means that for all $i \in \dom(w)$ there exists a subset $R_i \subseteq \sigma(Y)$ such that $(\code(w, \eta))_{\sigma[x/i][Y'/R_i]} \models \Phi_{Y}(\zeta')$. Then by induction hypothesis for all ${i \in \dom(w)}$ we have: $R_i \subseteq \dom(\eta)$ and $\llangle \zeta' \rrangle (w_{\sigma[x/i]}) = \eta|_{R_i}$. Let $R = \bigcup_{i \in \dom(w)} R_i$. Then, $(\code(w, \eta))_{\sigma[Z/R]} \models \xi(Z)$. Since $(\code(w, \eta))_{\sigma} \models \forall Z. (\xi(Z) \Rightarrow (Y \subseteq Z))$, we obtain $\sigma(Y) \subseteq R$. Hence $R = \sigma(Y)$ and
$$
\llangle \zeta \rrangle (w_{\sigma}) = \bigsqcap_{i \in \dom(w)} \eta|_{R_i} = \eta|_R = \eta|_{\sigma(Y)}.
$$
Finally, $\sigma(Y) = \bigcup_{i \in \dom(w)} R_i \subseteq \dom(\eta)$. This shows that the right hand side of (\ref{Eq:Techno}) holds true.
\item Conversely, assume that the right hand side of (\ref{Eq:Techno}) holds. Then there exists a family $(R_i)_{i \in \dom(w)}$ of subsets $R_i \subseteq \dom(Y) \subseteq \dom(\eta)$ such that $\bigcup_{i \in \dom(w)} R_i = \sigma(Y)$ and, for all $i \in \dom(w)$, $\llangle \zeta' \rrangle (w_{\sigma[x/i]}) = \eta|_{R_i}$. Then it is easy to see by induction hypothesis that, for all $i \in \dom(w)$, $(\code(w, \eta))_{[x/i][Y'/R_i]} \models \Phi_{Y'}(\zeta')$. Then $(\code(w, \eta))_{\sigma} \models \xi(Y)$. It remains to show that
$$
(\code(w, \sigma))_{\sigma} \models \forall Z. (\xi(Z) \Rightarrow (Y \subseteq Z)).
$$
Indeed, let $Q \subseteq \dom(w)$ with $(\code(w, \eta))_{\sigma[Z/Q]} \models \xi(Z)$. Then for all $i \in \dom(w)$ there exists a subset $Q_i \subseteq Q$ with $(\code(w, \eta))_{\sigma[x/i][Y'/Q_i]} \models \Phi_{Y'}(\zeta')$. Then by induction hypothesis for all $i \in \dom(w)$ we have ${Q_i \subseteq \dom(\eta)}$ and
$$
\eta|_{Q_i} = \llangle \zeta' \rrangle (w_{\sigma[x/i]}) = \eta|_{R_i}.
$$
Hence $Q_i = R_i$ for all $i \in \dom(w)$, and
$$
\sigma(Y) = \bigcup_{i \in \dom(w)} R_i = \bigcup_{i \in \dom(w)} Q_i \subseteq Q.
$$
\end{itemize}
\item The proof for $\zeta = {\sqcap} X. \zeta'$ with $X \in V_2$ is completely analogous to the proof of the previous case. The difference is that we consider "for all $I \subseteq \dom(w)$" instead of "for all $i \in \dom(w)$".
\end{itemize}

Finally, we construct $\Phi(\zeta)$ from $\Phi_Y(\zeta)$ by labelling all positions not in $Y$ by $\#$:
$${\Phi(\zeta) = \exists Y. (\Phi_Y(\zeta) \wedge \forall x. (Y(x) \vee \bigvee_{a \in \Sigma} P_{(a, \#)}(x)))}.$$
Assume that $\llangle \zeta \rrangle (w_{\sigma}) = \eta$. Let $R = \dom(\eta)$ and consider $\sigma' = \sigma[Y/R]$. Then $\sigma'(Y) \subseteq \dom(\eta)$ and $\llangle \zeta \rrangle (w_{\sigma}) = \eta|_{\sigma(Y)}$. Then by (\ref{Eq:Techno}) we have $(\code(w, \eta))_{\sigma'} \models \Phi_Y(\zeta)$. Moreover, for all $i \in \dom(w) \setminus \sigma'(Y)$, the value $\eta(i)$ is undefined and hence $(\code(w, \eta))_{\sigma'} \models \forall x. (Y(x) \vee \bigvee_{a \in \Sigma} P_{(a,\#)}(x))$ which implies $(\code(w, \eta))_{\sigma} \models \Phi(\zeta)$.
\qed
\end{proof}
Now we continue the proof of Lemma \ref{Lemma:UnDefRec}.  We apply Lemma \ref{Lemma:Aux} to the case ${\zeta = \varphi}$. Then, $\Phi(\varphi)$ is a sentence and $\mathcal L(\Phi(\varphi)) = \{\code(w, \eta) \; | \; \llangle \varphi \rrangle(w) = \eta \neq \bot \}$. Note that $\mathcal L(\Phi(\varphi))$ is $h$-unambiguous, since for every $w \in \Sigma^{\omega}$ there exists at most one $u \in \mathcal L(\Phi(\varphi))$ with $h(u) = w$. If we let $\beta = \Phi(\varphi)$, then we obtain the desired decomposition ${[\![\varphi]\!] = h((\val \circ g) \cap \mathcal L(\beta))}$.  Indeed, let $w \in \Sigma^{\omega}$. Then we distinguish between the following two cases:
\begin{itemize}
\item $\llangle \varphi \rrangle (w) = \bot$. Then $[\![\varphi]\!](w) = \zero$. On the other side, there exists no $\eta$ with $\code(w, \eta) \in \mathcal L(\beta)$ and hence no $u \in \mathcal L(\beta)$ with $h(u) = w$. Then ${h((\val \circ g) \cap \mathcal L(\beta))(w) = \zero = [\![\varphi]\!](w)}$.
\item $\llangle \varphi \rrangle (w) \in M^{\uparrow}$. Then, since the mapping $g$ assigns the default weight $\one$ to the undefined positions of $\llangle \varphi \rrangle (w) \in M^{\uparrow}$ and $\mathcal L(\beta)$ is $h$-unambiguous, we also have ${h((\val \circ g) \cap \mathcal L(\beta))(w) = [\![\varphi]\!](w)}$.

\end{itemize}

This finishes the proof of Lemma \ref{Lemma:UnDefRec}. Hence $\uWAL$-definability implies unambiguous recognizability.

\subsection{Unambiguous Case: Recognizability Implies Definability}

Now we show the converse part of Theorem \ref{Theorem:Main}(a), i.e., we show that unambiguous recognizability implies $\uWAL$-definability. 

\begin{lemma}
\label{Lemma:UnRecDef}
Let $\A$ be an unambiguous WBA over $\Sigma$ and $\mathbb V$. Then, the quantitative $\omega$-language $[\![\A]\!]$ if $\uWAL$-definable over $\mathbb V$.
\end{lemma}

\begin{proof}

Let $\A = (Q, I, T, F, \wt)$ be an unambiguous WBA over $\Sigma$ and $\mathbb V$. First, using the standard approach, we describe runs of $\A$ by means of MSO-formulas. For this, we fix an enumeration $(t_i)_{1 \le i \le m}$ of $T$ and associate with every transition $t_i$ a second-order variable $X_i$ which keeps track of positions where $t$ is taken. Then, a run of $\A$ can be described using a formula $\beta \in \MSO(\Sigma)$ with $\Free(\beta) = \{X_1, ..., X_m\}$ which demands that values of the variables $X_1, ..., X_m$ form a partition of the domain of an input word, the transitions of a run are matching, the labels of transitions of a run are compatible with an input word, a run starts in $I$ and visits some state in $F$ infinitely often. Let $\one \in M$ be an arbitrary default weight. Consider the  $\uWAL(\Sigma, \mathbb V_{\one})$-sentence
$$
\varphi = W(\exists X_1...  \exists X_m. \beta) \sqcap \big({\sqcap} X_1... {\sqcap} X_m. [W(\beta) \Rightarrow {\sqcap} x. {\textstyle \bigsqcap_{i = 1}^m} X_i(x) \Rightarrow (x \mapsto \wt(t_i))]\big).
$$
Now we show that $[\![\varphi]\!] = [\![\A]\!]$. Let $w \in \Sigma^{\omega}$.
We distinguish between the following two cases.
\begin{itemize}
\item $\Run_{\A}(w) = \emptyset$. Then $[\![\A]\!](w) = \zero$. On the other side, $w \nvDash \exists X_1... \exists X_m. \beta$ which implies $\llangle W(\exists X_1...\exists X_m. \beta) \rrangle(w) = \bot$. Then $\llangle \varphi \rrangle (w) = \bot$ and hence ${[\![\varphi]\!](w) = \zero = [\![\A]\!](w)}$.
\item $\Run_{\A}(w) \neq \emptyset$. Since $\A$ is unambiguous, we have $\Run_{\A}(w) = \{\rho\}$. Let $\rho = (\tau_i)_{i \in \N}$ and $\sigma$ be a fixed $w$-assignment. Then, there exists exactly one tuple $(I_1, ..., I_m) \in (2^{\dom(w)})^m$ such that 
$w_{\sigma[X_1/I_1]...[X_m/I_m]} \models \beta$. Then $\llangle W(\exists X_1...\exists X_m.\beta) \rrangle(w) = \top$. Moreover,
$$\llangle W(\beta) \Rightarrow {\sqcap} x. {\textstyle \bigsqcap_{i = 1}^m} X_i(x) \Rightarrow (x \mapsto \wt(t_i)) \rrangle (w_{\sigma[X_1/I_1]...[X_m/I_m]}) = (\wt(\tau_i))_{i \in \N}$$
and, for all $(J_1, ..., J_m) \in (2^{\dom(w)})^m$ with $(J_1, ..., J_m) \neq (I_1, ..., I_m)$, we have
$$\llangle W(\beta) \Rightarrow {\sqcap} x. {\textstyle \bigsqcap_{i = 1}^m} X_i(x) \Rightarrow (x \mapsto \wt(t_i)) \rrangle (w_{\sigma[X_1/J_1]...[X_m/J_m]}) = \top.$$
Then $\llangle \varphi \rrangle (w) = (\wt(\tau_i))_{i \in \N}$ and hence $[\![\varphi]\!](w) = \wt_{\A}(\rho) = [\![\A]\!](w).$
\end{itemize}
Hence $[\![\A]\!]$ is $\uWAL$-definable over $\mathbb V$.
\qed
\end{proof}

\subsection{Nondeterministic Case: Definability Implies Recognizability}

Now we turn to the proof of Theorem \ref{Theorem:Main}(b). First we show that $\eWAL$-definability implies nondeterministic recognizability. 

\begin{lemma}
Let $\one \in M$ be a default weight and $\psi \in \eWAL(\Sigma, \mathbb V_{\one})$. Then the quantitative $\omega$-language $[\![\varphi]\!]$ is recognizable over $\mathbb V$.
\end{lemma}

\begin{proof}
The idea of our proof is similar to the unambiguous case, i.e., via a decomposition of the $\eWAL$-sentence $\psi$. 
We show that there exist an extended alphabet $\Gamma$, renamings $h: \Gamma \to \Sigma$ and $g: \Gamma \to M$, and a sentence $\beta \in \MSO(\Gamma)$ such that $[\![\varphi]\!] = h((\val \circ g) \cap \mathcal L(\beta))$.
Note that, as opposed to the unambiguous case, the $\omega$-language $\mathcal L(\beta)$ is not necessarily $h$-unambiguous.

We may assume that ${\psi = {\sqcup} x_1... {\sqcup} x_k. {\sqcup} X_1... {\sqcup} X_l. \varphi}$ where $\varphi \in \uWAL(\Sigma, \mathbb V_{\one})$ and $x_1, ..., x_k$, $X_1$, ..., $X_l$ are pairwise distinct variables. 

As opposed to the unambiguous case, the extended alphabet $\Gamma$ must also keep track of the values of the variables $x_1, ..., x_k, X_1, ..., X_l$. Let $\mathcal V = \{x_1, ..., x_k, X_1, ..., X_l\}$ and $\Delta_{\varphi}$ be defined as in the unambiguous case. 
Then we let $\Gamma = \Sigma \times \Delta_{\varphi} \times 2^{\mathcal V}$ and define $h, g$ for all $u = (a, b, S) \in \Gamma$ with $a \in \Sigma$, $b \in \Delta_{\varphi}$ and $S \subseteq \mathcal V$ by $h(u) = a$ and $g(u) = b$ if $b \in M$ and $g(u) = \one$ otherwise. Finally we construct the MSO-sentence $\beta$ over $\Gamma$.
The construction of $\beta$ will be based on Lemma \ref{Lemma:Aux}. Let $\Phi(\varphi) \in \MSO(\Sigma \times \Delta_{\varphi})$ be the formula constructed in Lemma \ref{Lemma:Aux} for $\zeta = \varphi$. 
Let $\overline{\Phi(\varphi)} \in \MSO(\Gamma)$ be the formula obtained from $\Phi(\varphi)$ by replacing every predicate $P_{(a, b)}(x)$ occurring in $\Phi(\varphi)$ by the formula ${\bigvee (P_{(a, b, U)}(x) \; | \; U \subseteq \mathcal V)}$. Using the standard B\"uchi encoding technique we construct the formula $\phi \in \MSO(\Gamma)$ which encodes the values of $\mathcal V$-variables in the $2^{\mathcal V}$-component of an $\omega$-word over $\Gamma$. We let $\phi = \forall y. (\phi_1 \wedge \phi_2)$ where
\begin{align*}
\phi_1 &= \bigwedge_{x \in \mathcal V \cap V_1} ([R_{x, 1}(y) \wedge (y = x)] \vee [R_{x, 0}(y) \wedge (y \neq x)]), \\
\phi_2 &= \bigwedge_{X \in \mathcal V \cap V_2} ([R_{X, 1}(y) \wedge X(y)] \vee [R_{X, 0}(y) \wedge \lnot X(y)])
\end{align*}
and, for $\mathcal X \in V$ and $i \in \{0,1\}$, $R_{\mathcal X, i}(y)$ denotes the formula
$$
\bigvee (P_{(a,b,S)}(y) \; | \; a \in \Sigma, b \in \Delta_{\varphi} \text{ and } S \subseteq \mathcal V \text{ with } \mathcal X \triangleleft_i S)
$$
where ${\triangleleft}_1 = {\in}$ and ${\triangleleft}_0 = {\notin}$.

Then we let $\beta = \exists x_1...\exists x_k.\exists X_1...\exists X_l. (\phi \wedge \overline{\Phi(\varphi)})$. It remains to show that ${[\![\psi]\!] = h((\val \circ g) \cap \mathcal L(\beta))}$.  

Let $w = (a_i)_{i \in \N} \in \Sigma^{\omega}$. For any $u = (b_i)_{i \in \N} \in \Delta_{\varphi}^{\omega}$ we will abuse notation and write $(w, u)$ for $((a_i, b_i))_{i \in \N}$. For $w \in \Sigma^{\omega}$, let $\mathcal V_{w}$ denote the set of all mappings $\mathcal J: \mathcal V \to \dom(w) \cup 2^{\dom(w)}$ such that $\mathcal J(\mathcal V \cap V_1) \subseteq \dom(w)$ and $\mathcal J(\mathcal V \cap V_2) \subseteq 2^{\dom(w)}$. For a $w$-assignment $\sigma$ and $\mathcal J \in \mathcal V_w$, let $\sigma' := \sigma[\mathcal V/\mathcal J]$ denote the $w$-assignment such that $\sigma'|_{\mathcal V} = \mathcal J$ and $\sigma'_{V \setminus \mathcal V} = \sigma|_{V \setminus \mathcal V}$. Then
\begin{align*}
h((\val \circ g) \cap \mathcal L(\beta))(w) &= \sum (\val(g(u)) \; | \; \mathcal J \in \mathcal V_w \text{ and } (w, u)_{\sigma[\mathcal V/\mathcal J]} \models \Phi(\varphi)) \\
&\overset{(!)}{=} \sum_{\mathcal J \in \mathcal V_w} [\![\varphi]\!](w_{\sigma[\mathcal V/\mathcal J]}) \\
& = [\![\psi]\!](w).
\end{align*}
Then, the quantitative $\omega$-language $[\![\psi]\!]$ is recognizable over $\mathbb V$ by Theorem \ref{Theorem:Nivat} (b) and the classical B\"uchi theorem (which states that $\mathcal L(\beta)$ is a recognizable $\omega$-language). 
\qed
\end{proof}

\section{Nondeterministic Case: Recognizability implies Definability}

Now we show the converse direction of Theorem \ref{Theorem:Main}(b), i.e., that recognizability implies $\eWAL$-definability. 

\begin{lemma}
Let $\A$ be a WBA over $\Sigma$ and $\mathbb V$. Then the quantitative $\omega$-language $[\![\A]\!]$ is $\eWAL$-definable over $\mathbb V$.
\end{lemma}

\begin{proof}
Our proof is a slight modification of our proof of Lemma \ref{Lemma:UnRecDef}. Let ${\A = (Q, I, T, F, \wt)}$ be a nondeterministic WBA. Adopting the notations from the proof of Lemma \ref{Subsect:Unamb_Case}, we construct the $\eWAL(\Sigma, \mathbb V_{\one})$-sentence
$$
\psi = {\sqcup} X_1... {\sqcup} X_m. \big(W(\beta) \Rightarrow {\sqcap} x. {\textstyle \bigsqcap_{i = 1}^m} X_i(x) \Rightarrow (x \mapsto \wt(t_i)) \big).
$$
(where $\one$ is irrelevant for the definition of $\psi$). Now we show that $[\![\varphi]\!] = [\![\A]\!]$. Let $w \in \Sigma^{\omega}$. Then, using the correspondence between the values of $X_1, ..., X_m$ and the runs in $\Run_{\A}(w)$, we obtain
$$
[\![\psi]\!](w) = \sum_{\rho = (\tau_i)_{i \in \N} \in \Run_{\A}(w)} \val(\wt(\tau_i)) = \sum_{\rho \in \Run_{\A}(w)} \wt_{\A}(\rho) = [\![\A]\!](w).
$$
This shows that $[\![\psi]\!]$ is $\eWAL$-definable over $\mathbb V$.
\qed
\end{proof}

\section{Discussion}

In this paper we introduced a weight assignment logic which is a simple and intuitive logical formalism for reasoning about quantitative $\omega$-languages. Moreover, it works with arbitrary valuation functions whereas in weighted logics of \cite{DM12}, \cite{DP14} some additional restrictions on valuation functions were added. We showed that WAL is expressively equivalent to unambiguous weighted B\"uchi automata. We also considered an extension of WAL which is equivalent to nondeterministic B\"uchi automata. Our expressiveness equivalence results can be helpful to obtain decidability properties for our new logics. The future research should investigate decidability properties of nondeterministic and unambiguous weighted B\"uchi automata with the practically relevant objectives. Although the weighted $\omega$-automata models \cite{CDH08} do not have a B\"uchi acceptance condition, it seems likely that their decidability results about the threshold problems hold for 
B\"uchi acceptance condition as well. It could be also interesting to study our weight assignment technique in the context of temporal logic like LTL.

Our results obtained for $\omega$-words can be easily adopted to the structures like finite words and trees. We have also extended the results of this paper to the timed setting and obtained a logical characterization of {\em multi-weighted timed automata} (cf., e.g., \cite{BBL08}, \cite{LR05}). For the proof of this result we applied a Nivat decomposition theorem for weighted timed automata \cite{DP14}. Due to space constraints we cannot present this result here.

\end{document}